\pdfoutput=1
\documentclass[11pt,a4paper]{article}

\usepackage[margin=1in]{geometry}
\usepackage{amsthm}

\newif\ifarxiv
\arxivtrue

\usepackage[T1]{fontenc}
\usepackage{graphicx}

\usepackage{amsmath}
\usepackage{amssymb}
\usepackage{graphbox}
\usepackage{subcaption}
\usepackage{enumitem}
\setlist{noitemsep, topsep=4pt}
\usepackage{thm-restate}
\usepackage{macros}

\usepackage{hyperref}
\usepackage[capitalise,noabbrev]{cleveref}

\usepackage{xspace}
\newcommand{\vAND}{\textup{\textsc{and}}\xspace}
\newcommand{\vOR}{\textup{\textsc{or}}\xspace}
\newcommand{\vMAJ}{\textup{\textsc{maj}}\xspace}
\newcommand{\pdfscale}{0.75}

\newcommand{\PlBiREMCC}{\textsc{Planar Bipartite \mbox{$(3, 1)_\mathrm{R}$}-\mbox{RE-MCC}}\xspace}
\newcommand{\PlBiMCC}{\textsc{Planar Bipartite \mbox{$(3, \le 2)$}-\mbox{MCC}}\xspace}

\newcommand{\Kakuro}{\textsc{Kakuro}\xspace}

\newcommand{\Lmin}{L_\mathrm{min}}
\newcommand{\Lmax}{L_\mathrm{max}}
\NewDocumentCommand{\GenKakuro}{O{N} O{\Lmin} O{\Lmax}}{%
  \mbox{$(#1, #2, #3)$}-\Kakuro
}

\crefname{equation}{}{}
\Crefname{equation}{Equation}{Equations}

\theoremstyle{plain}
\newtheorem{theorem}{Theorem}

\newtheorem{corollary}{Corollary}

\theoremstyle{definition}
\newtheorem{definition}{Definition}
\theoremstyle{remark}
\newtheorem{remark}{Remark}

\newenvironment{proofof}[1]{\begin{proof}[Proof of #1]}{\end{proof}}

\usepackage[section,below]{placeins}

\makeatletter
\def\fps@figure{tbp}
\makeatother

\title{Mixed Cycle Cover Problem with Required-Edge Constraints}
\author{%
  Kosuke Susukita\thanks{University of Hyogo, Japan. \mailto{ad25x034@guh.u-hyogo.ac.jp}}
  \and
  Junichi Teruyama\thanks{University of Hyogo, Japan. \mailto{junichi.teruyama@gsis.u-hyogo.ac.jp}}
}
\date{}

\begin{document}

\maketitle

\begin{abstract}
  In this paper, we develop a unified framework for establishing the ASP-completeness (NP-completeness under parsimonious reductions) of pencil-and-paper puzzles, based on restricted cycle covers.
  We introduce the \textsc{Required-edge Mixed Cycle Cover} problem (RE-MCC) and prove that both RE-MCC and standard MCC are ASP-complete on planar bipartite degree-$3$ graphs.
  Exploiting their rigid ternary structure, we resolve an open problem posed by the MIT Hardness Group (2024) by proving the ASP-completeness of planar \textsc{Constraint Graph Satisfiability} (CGS) with matching edge weights.
  We also completely settle the computational complexity dichotomy of \textsc{Kakuro} regarding the maximum available digit and the maximum lengths of contiguous blank cells.
  Finally, we translate RE-MCC into an equivalent half-integral orientation model on grid graphs.
  Because this model enables dense gadget tiling without unconstrained blank regions, it systematically yields ASP-completeness for various regional puzzles, including \textsc{Chocona}, \textsc{Four Cells}, \textsc{Hinge}, and \textsc{Shimaguni}, while strengthening existing hardness results for \textsc{Choco Banana} and \textsc{Five Cells}.
\end{abstract}

\section{Introduction}

In the study of pencil-and-paper puzzles, computational hardness is typically established via reductions from variants of Boolean satisfiability (SAT).
However, while SAT instances are purely combinatorial, grid-based puzzles are subject to rigid geometric, arithmetic, or regional constraints.
Furthermore, because solution uniqueness is a fundamental aesthetic and practical requirement in puzzles, reductions must strictly preserve the number of solutions.
Such reductions are called \emph{parsimonious} reductions, and NP-completeness under these reductions is referred to as \emph{ASP-completeness}~\cite{YatoSeta03ASP}.
While many well-known SAT variants are ASP-complete (e.g., 3-SAT~\cite{Valiant1979Permanent,YatoSeta03ASP} and \textsc{Planar 3-SAT}~\cite{Hunt98PlanarCounting}), constructing parsimonious puzzle gadgets directly from SAT poses a significant challenge due to the geometric rigidity of puzzle boards.

To systematically establish the ASP-completeness of puzzles, a framework based on the \textsc{Hamiltonian Cycle} Problem on grid graphs~\cite{MITHardness24Hamiltonicity,Tang22Loop} is explored.
While this framework is effective for loop-drawing puzzles, it relies heavily on global connectivity constraints.
Alternatively, the \emph{Constraint Graph Satisfiability} (CGS) framework~\cite{HearnDemaine09Puzzles} provides a versatile model based on graph orientations.
However, the ASP-completeness for the most restrictive  setting---with matching edge weights---has remained open~\cite{MITHardness24GeneralizedSAT}.

\subsection{Our Contributions}

In this paper, we establish a comprehensive framework connecting abstract cycle covers to concrete grid puzzles.
Our contributions are organized into three main parts: theoretical hardness foundations, direct applications of some complexity results, and convenient geometric frameworks for puzzle reductions.

Specifically, our framework is based on a variant of the \textsc{Mixed Cycle Cover} problem (MCC), which asks whether a given mixed graph (i.e., a graph with both directed and undirected edges) can be covered by vertex-disjoint cycles.
It is well known that cycle covers on directed and undirected graphs can be computed in polynomial time.
However, finding cycle covers becomes ASP-complete when mixed graphs are considered~\cite{Seta02CrossSum}.
This problem was originally analyzed to study the ASP-completeness of Kakuro (Cross Sum) puzzles, and was shown to be ASP-complete even when restricted to planar max-degree-$3$ graphs~\cite{Seta02CrossSum}.
To make the problem more suitable for puzzle reductions, we introduce \emph{required edges}---edges that must be included in the cycle cover---into MCC, and define this problem as the \textsc{Required-edge Mixed Cycle Cover} problem (RE-MCC).
We prove that RE-MCC is ASP-complete even when instances are restricted to planar degree-$3$ bipartite graphs where each vertex is incident to exactly one required edge.
Moreover, by constructing a parsimonious gadget simulating a required edge, we obtain a corollary that MCC is ASP-complete even when instances are restricted to planar degree-$3$ bipartite mixed graphs where each vertex is incident to at most two undirected edges.

Second, we establish the ASP-completeness of several problems via reductions from MCC under the above restriction.
Crucial to these reductions is the rigid structure of this problem:
each vertex admits an exclusive ternary configuration (one incoming, one outgoing, and one unused edge), and the bipartiteness allows for asymmetric edge gadgets.
Applying this strategy, we prove the ASP-completeness of CGS with matching edge weights using \vAND, \vOR, and \vMAJ vertices, thereby resolving the open problem mentioned above~\cite{MITHardness24GeneralizedSAT}.
We also strengthen the known ASP-completeness of \textsc{Kakuro} by \cite{Ruepp10Kakuro,Seta02CrossSum,YatoSeta03ASP}.
Specifically, we show that \textsc{Kakuro} remains ASP-complete even when the available digits are restricted to $\{1, 2, 3\}$.
Combined with previous results, this completes the complexity dichotomy of \textsc{Kakuro} with respect to two fundamental parameters: the maximum available digit and the maximum length of consecutive blank cells.

Finally, to bridge the gap between abstract graphs and grid-based puzzles, we introduce an \emph{orientation model} equivalent to RE-MCC under our restriction.
This model consists of a rectangular grid graph with vertex in-degree conditions and four types of edge constraints.
We can construct such a grid graph such that its valid half-integral orientations correspond one-to-one with the cycle covers of a given RE-MCC instance.
Conceptually, each edge distributes weight $1$ to its endpoints in increments of $1/2$ and each vertex has a uniform in-degree; these constraints map naturally to puzzle rules involving regional or arithmetic (e.g., sum) constraints.
By aligning these gadgets on a rectangular grid, we can tile the puzzle board with gadgets in a systematic manner, thereby directly translating the combinatorial hardness of the RE-MCC into a rigid geometric configuration.

Using this framework, we establish new ASP-completeness results for several puzzles, including Chocona, Four Cells, Hinge, and Shimaguni.
Moreover, we strengthen existing NP-completeness results for other puzzles, specifically Choco Banana~\cite{Iwamoto2024ChocoBanana} and Five Cells~\cite{Iwamoto2022FiveCells}, to ASP-completeness.

\section{Preliminaries}

\begin{definition}[\cite{YatoSeta03ASP}]
  A search problem is \emph{ASP-complete} if it is NP-complete under polynomial-time reductions such that, for each instance, there exists a polynomial-time computable bijection between the solution sets of the original and reduced instances.
\end{definition}

ASP-completeness implies that, for any fixed positive integer $k$, it is NP-complete to decide whether, given an instance and $k$ solutions, there exists a $(k+1)$-th solution~\cite{YatoSeta03ASP}.
Moreover, because such reductions are \emph{parsimonious}, that is, they preserve the exact number of solutions, every ASP-complete problem is inherently \#P-complete.
(As is standard, our parsimonious reductions naturally yield polynomial-time computable bijections between solutions.)
By a slight abuse of terminology, we use the term ``ASP-complete'' to refer to both the decision and counting versions of the same problem.

\begin{definition}
  A \emph{mixed graph} $G = (V, E, A)$ is a graph consisting of a set of vertices $V$, a set of undirected edges $E$, and a set of directed edges $A$.
  The \emph{total degree} of a vertex $v$ is defined as the number of undirected edges incident to $v$ plus the number of directed edges incident to $v$ (i.e., having $v$ as a head or tail).
\end{definition}

A (mixed) cycle in a mixed graph is defined as a sequence of vertices $v_1, \allowbreak v_2, \dots, \allowbreak v_{k+1}$ with length larger than 3 such that $v_1, v_2, \dots, v_k$ are pairwise distinct, $v_1 = v_{k+1}$ and for each $i$, $\{v_i, v_{i+1}\} \in E$ or $(v_i, v_{i+1}) \in A$.
\footnote{Two cycles are considered the same if they differ only in their starting vertex (their orientations must also be the same).}

A set of cycles $\mathcal{C}$ of a graph $G$ is called a \emph{cycle cover} of $G$ if every vertex of $G$ is contained in exactly one cycle in $\mathcal{C}$.
\textsc{Mixed Cycle Cover} (MCC) asks whether a given mixed graph has a cycle cover.

\begin{definition}
  Given a mixed graph $G = (V, E, A)$ and a set of undirected edges $R \subseteq E$, \emph{\textsc{Required-Edge Mixed Cycle Cover}} (RE-MCC) asks whether there exists a cycle cover of $G$ that includes all edges in $R$.
  The edges in $R$ are called \emph{required edges}.
\end{definition}

We note that RE-MCC is a generalization of MCC.

MCC/RE-MCC where the input graph is planar is called \textsc{Planar MCC/\allowbreak RE-MCC}.
In similar fashion, MCC/RE-MCC where the input graph is bipartite is called \textsc{Bipartite MCC/RE-MCC}.
MCC where each vertex has total degree $d$ and undirected degree $d'$ is called \textsc{$(d, d')$-MCC}.
RE-MCC where each vertex has total degree $d$ and degree of required edges $d'$ is called \textsc{$(d, d')_\mathrm{R}$-RE-MCC}.
In both cases, if we are given a maximum degree $\Delta$ instead of a fixed degree $d$, we denote ``$\le \Delta$'' instead of ``$d$''.

Seta~\cite{Seta02CrossSum} proved the following theorem on the complexity of MCC.

\begin{theorem}[\cite{Seta02CrossSum}]\label{thm:seta_mcc}
  \textsc{Planar $(\le 3, \le 3)$-MCC} is ASP-complete.
\end{theorem}

\section{\textsc{Required-Edge Mixed Cycle Cover}}\label{sec:RE-MCC}

In this section, we prove the following theorem, which is the main result of this paper.

\begin{theorem}\label{thm:main}
  \PlBiREMCC is ASP-complete.
\end{theorem}

\begin{proof}
  We reduce from \textsc{Planar Positive 1-in-3-SAT}.
  In this problem, the incidence graph is planar, every clause consists of three positive literals, and we ask for a truth assignment where exactly one literal per clause is true.
  It is parsimoniously reducible from \textsc{Planar 3-SAT}~\cite{Hunt98PlanarCounting}, which implies that the problem is ASP-complete.

  \Cref{fig:RE-MCC_variable,fig:RE-MCC_clause} show the gadgets used in the reduction.
  Vertices are colored black and white to indicate the two independent sets.
  Required edges are shown as bold lines.
  \Cref{fig:RE-MCC_variable} illustrates an example of the variable gadget for a variable with three occurrences.
  The gadget admits exactly two valid cycle covers. Covering the external edges corresponds to assigning true to the variable (\cref{fig:RE-MCC_variable_true}), whereas not covering them corresponds to assigning false (\cref{fig:RE-MCC_variable_false}).
  The gadget can be extended to match the number of occurrences of the variable in the formula.
  \Cref{fig:RE-MCC_clause} illustrates the clause gadget.
  This gadget admits exactly three valid cycle covers, each of which corresponds to an assignment where exactly one literal is true.
  We arrange the variable and clause gadgets according to a planar embedding of the given \textsc{Planar Positive 1-in-3-SAT} instance, connecting their external edges.
  This construction introduces no edge crossings, preserving the planarity of the original instance.
  The resulting reduction is parsimonious; thus the theorem holds.
\end{proof}

\begin{figure}
  \begin{subfigure}[t]{0.32\linewidth}
    \centering
    \includegraphics{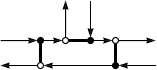}
    \caption{Structure}
  \end{subfigure}
  \hfill
  \begin{subfigure}[t]{0.32\linewidth}
    \centering
    \includegraphics{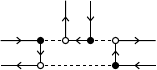}
    \caption{True assignment}\label{fig:RE-MCC_variable_true}
  \end{subfigure}
  \hfill
  \begin{subfigure}[t]{0.32\linewidth}
    \centering
    \includegraphics{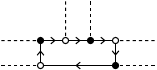}
    \caption{False assignment}\label{fig:RE-MCC_variable_false}
  \end{subfigure}
  \caption{
    An example of a variable gadget and its two valid cycle covers.
    Each pair of external edges connects to a clause gadget.
    When the variable is assigned true, the cycles cover the external edges.
    When the variable is assigned false, the cycles do not cover the external edges.
    Although this example features three connections, the gadget can be extended to match the number of occurrences of the variable in the formula.
  }
  \label{fig:RE-MCC_variable}
\end{figure}

\begin{figure}
  \begin{subfigure}{\linewidth}
    \centering
    \includegraphics{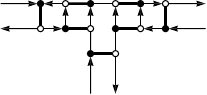}
  \end{subfigure}\vspace{1em}

  \begin{subfigure}{0.32\linewidth}
    \centering
    \includegraphics{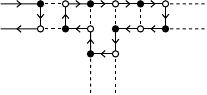}
  \end{subfigure}
  \hfill
  \begin{subfigure}{0.32\linewidth}
    \centering
    \includegraphics{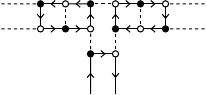}
  \end{subfigure}
  \hfill
  \begin{subfigure}{0.32\linewidth}
    \centering
    \includegraphics{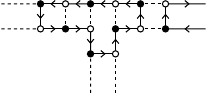}
  \end{subfigure}
  \caption{
    Top row: the clause gadget.
    Bottom row: its three valid cycle covers.
    Each pair of external edges connects to a variable gadget.
    In each valid cycle cover, exactly one pair of external edges is covered, corresponding to the assignment of exactly one true literal in the clause.
  }
  \label{fig:RE-MCC_clause}
\end{figure}

Furthermore, a required edge can be parsimoniously simulated by the gadget shown in \cref{fig:RE-MCC_required_edge_simulation}, and this substitution preserves bipartiteness.
Since every vertex is incident to at most two undirected edges in our construction (after this modification), the following corollary also holds, which is a strengthening of \cref{thm:seta_mcc}.

\begin{corollary}\label{cor:MCC}
  \PlBiMCC is ASP-complete.
\end{corollary}

\begin{figure}
  \begin{subfigure}{0.32\linewidth}
    \centering
    \includegraphics{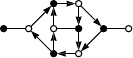}
  \end{subfigure}\hfill
  \begin{subfigure}{0.32\linewidth}
    \centering
    \includegraphics{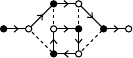}
  \end{subfigure}\hfill
  \begin{subfigure}{0.32\linewidth}
    \centering
    \includegraphics{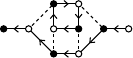}
  \end{subfigure}
  \caption{Simulation of a required edge using undirected and directed edges.}
  \label{fig:RE-MCC_required_edge_simulation}
\end{figure}

\ifarxiv
\begin{remark}
  It is known that, on general mixed graphs, the problems of finding a cycle cover consisting of only even cycles and of finding one containing at least one even cycle are both NP-complete~\cite{Huasrsch2025CycleFactors}.
  Since every cycle in a bipartite graph is even, both problems coincide with standard MCC in our bipartite setting; thus, \cref{cor:MCC} immediately strengthens these NP-completeness results to ASP-completeness.
\end{remark}
\fi

\section{Reductions from \textsc{Mixed Cycle Cover}}\label{sec:reductions_from_MCC}

In this section, we establish the ASP-completeness of two problems, \textsc{Constraint Graph Satisfiability} (CGS) and \Kakuro, via parsimonious reductions from \PlBiMCC.
Crucial to both reductions is the rigid local structure of \PlBiMCC: because every vertex has degree~$3$, fixing a cycle cover forces an exclusive ternary configuration (exactly one incoming, one outgoing, and one unused edge).
From the edge perspective, each undirected edge takes one of three states (traversed in either direction or unused), while directed edges admit two states.
Furthermore, exploiting the bipartiteness of the graph, we can employ two complementary types of vertex gadgets connected by asymmetric edge gadgets.

\subsection{\textsc{Constraint Graph Satisfiability}}\label{sec:cgs}

\begin{definition}[\cite{HearnDemaine09Puzzles}]
  A \emph{constraint graph} is a graph with vertex weights and edge weights.
  Each edge can be oriented in either direction.
  A vertex is \emph{satisfied} if and only if the sum of the weights of its incoming edges is at least the vertex weight.
  Given a constraint graph, the \textsc{Constraint Graph Satisfiability} problem (CGS) asks whether there exists an orientation of the edges that satisfies all vertices.
\end{definition}

CGS for planar graphs is called \textsc{Planar CGS}.
Typically, all vertex weights are restricted to $2$, and all edge weights are restricted to either $1$ or $2$.
Each edge is colored according to its weight: weight-$1$ edges are \emph{red} (represented by thin lines), and weight-$2$ edges are \emph{blue} (represented by thick lines).

We employ three standard types of vertices which are commonly used in constraint logic frameworks~\cite{HearnDemaine09Puzzles}:
\begin{description}
  \item[and] Incident to one blue edge and two red edges.
        The blue edge can be oriented outward only if both red edges are oriented inward.
  \item[or] Incident to three blue edges.
        At least one of the three edges must be oriented inward.
  \item[maj] Incident to three red edges.
        At least two of the three edges must be oriented inward.
\end{description}
All valid orientations of these vertices are illustrated in \cref{fig:cgs_vertices}.

\begin{figure}
  \centering
  \begin{tabular}{r c ccc ccc c}
    \vAND: & \includegraphics[align=c]{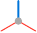}
           & \includegraphics[align=c]{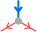}
           &
           &
           & \includegraphics[align=c]{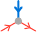}
           & \includegraphics[align=c]{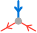}
           & \includegraphics[align=c]{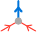}
           & \includegraphics[align=c]{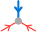} \\[1em]

    \vOR:  & \includegraphics[align=c]{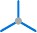}
           & \includegraphics[align=c]{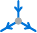}
           & \includegraphics[align=c]{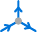}
           & \includegraphics[align=c]{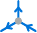}
           & \includegraphics[align=c]{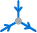}
           & \includegraphics[align=c]{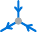}
           & \includegraphics[align=c]{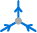}
           & \includegraphics[align=c]{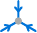}  \\[1em]

    \vMAJ: & \includegraphics[align=c]{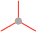}
           &
           &
           &
           & \includegraphics[align=c]{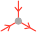}
           & \includegraphics[align=c]{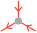}
           & \includegraphics[align=c]{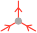}
           & \includegraphics[align=c]{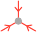} \\
  \end{tabular}
  \caption{
    Three types of vertices in constraint graphs and all of their valid orientations.
    Red and blue edges are depicted as thin light-gray and thick dark-gray lines, respectively.
  }
  \label{fig:cgs_vertices}
\end{figure}

While \textsc{Planar CGS} with only \vAND and \vOR vertices is NP-complete~\cite{HearnDemaine09Puzzles}, extending this to ASP-completeness is nontrivial because the known reduction relies on a non-parsimonious gadget that converts red and blue edges~\cite{HearnDemaine09Puzzles,MITHardness24GeneralizedSAT}.
To address this, recent work~\cite{MITHardness24GeneralizedSAT} formalized two CGS variants: \emph{arbitrary edge weights} (allowing different weights at endpoints) and standard \emph{matching edge weights}.
While CGS with arbitrary edge weights using \vAND, \vOR, and \vMAJ vertices was shown to be ASP-complete even on planar graphs~\cite{MITHardness24GeneralizedSAT}, the complexity for matching edge weights has remained open.
In this section, we resolve this open problem by establishing the following theorem.

\begin{restatable}{theorem}{thmcgs}
  \label{thm:cgs}
  \textsc{Planar CGS} with matching edge weights using only \vAND, \vOR, and \vMAJ vertices is ASP-complete.
\end{restatable}

We construct a reduction from \PlBiMCC.
\Cref{fig:cgs_gadget_connected} shows the idea of the reduction, which consists of two main components: vertex gadgets~(\cref{fig:cgs_vertex}) and edge gadgets~(\cref{fig:cgs_edge}).
The two types of vertex gadgets correspond to the two independent sets of the bipartite graph, and the three types of edge gadgets correspond to the three types of edges in \PlBiMCC: undirected and two directed edges (one for each direction).
Each terminal of a gadget has three edges labeled $I$, $O$, and $U$, which are connected to the corresponding labels of the other gadget.
The orientation of these edge triples determines the state of the gadget, which collectively simulates cycle covers of the original \PlBiMCC instance.
Specifically, exactly one of an edge triple oriented from a red vertex to a blue vertex represents the state $(I, U, O)$, which respectively correspond to an incoming edge, an unused edge, and an outgoing edge.
\ifarxiv\else
The crossing gadgets to ensure planarity are shown in \cref{fig:cgs_crossing_red_red,fig:cgs_crossing_red_blue,fig:cgs_crossing_blue_blue}; see the full version of this paper for the complete proof of correctness.
\fi

\begin{figure}
  \begin{subfigure}[c]{0.48\linewidth}
    \centering
    \includegraphics[scale=3]{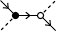}
  \end{subfigure}
  \hfill
  \begin{subfigure}[c]{0.48\linewidth}
    \centering
    \includegraphics{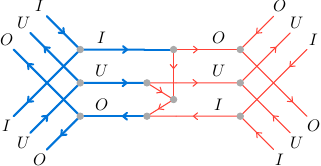}
  \end{subfigure}
  \caption{
    The idea of the reduction from \PlBiMCC to \textsc{Planar CGS}.
    The left figure shows a part of a \PlBiMCC instance, and the right figure shows the corresponding part of a \textsc{Planar CGS} instance.
  }
  \label{fig:cgs_gadget_connected}
\end{figure}

\begin{figure}
  \begin{subfigure}{0.48\linewidth}
    \centering
    \includegraphics{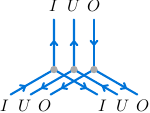}
    \caption{Blue}\label{fig:cgs_vertex_blue}
  \end{subfigure}\hfill
  \begin{subfigure}{0.48\linewidth}
    \centering
    \includegraphics{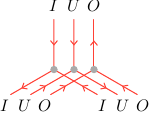}
    \caption{Red}\label{fig:cgs_vertex_red}
  \end{subfigure}
  \caption{
    Two types of vertex gadgets for CGS and an example of a valid orientation of their edges.
    The gadgets have three terminals, and each terminal consists of three edges labeled $I, U, O$.
    In each terminal, the single edge oriented differently from the other two determines the state of the terminal: $(O, U, I)$ in (a) and (b), as viewed clockwise from the top.
  }
  \label{fig:cgs_vertex}
\end{figure}

\begin{figure}
  \begin{subfigure}{0.32\linewidth}
    \centering
    \includegraphics{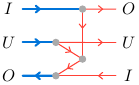}
    \caption{Undirected}\label{fig:cgs_edge_undirected}
  \end{subfigure}\hfill
  \begin{subfigure}{0.32\linewidth}
    \centering
    \includegraphics{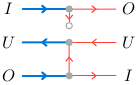}
    \caption{Directed (leftward)}\label{fig:cgs_edge_directed_left}
  \end{subfigure}\hfill
  \begin{subfigure}{0.32\linewidth}
    \centering
    \includegraphics{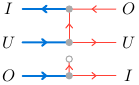}
    \caption{Directed (rightward)}\label{fig:cgs_edge_directed_right}
  \end{subfigure}
  \caption{
    Three types of edge gadgets for CGS and an example of a valid orientation of their internal edges.
    The hollow vertices indicate weight-1 vertices, which can be easily simulated using the standard vertices (see \cref{fig:cgs_weight1}).
    Exactly one of the three parallel paths oriented leftward determines the state of the gadget: $(O, I)$ in (a), $(U, U)$ in (b), and $(I, O)$ in (c), respectively.
  }
  \label{fig:cgs_edge}
\end{figure}

\begin{figure}
  \begin{subfigure}[b]{0.30\linewidth}
    \centering
    \includegraphics[scale=0.95]{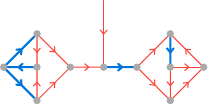}
    \caption{Weight-1 vertex}\label{fig:cgs_weight1}
  \end{subfigure}\hfill
  \begin{subfigure}[b]{0.29\linewidth}
    \centering
    \includegraphics[scale=0.95]{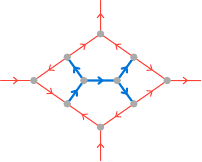}
    \caption{Red-red}\label{fig:cgs_crossing_red_red}
  \end{subfigure}\hfill
  \begin{subfigure}[b]{0.195\linewidth}
    \centering
    \includegraphics[scale=0.95]{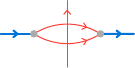}
    \caption{Red-blue}\label{fig:cgs_crossing_red_blue}
  \end{subfigure}\hfill
  \begin{subfigure}[b]{0.195\linewidth}
    \centering
    \includegraphics[scale=0.95]{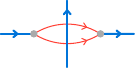}
    \caption{Blue-blue}\label{fig:cgs_crossing_blue_blue}
  \end{subfigure}
  \caption{
    Auxiliary gadgets for CGS: (a)~simulation of a weight-1 vertex using \vAND, \vOR, and \vMAJ vertices (showing unique valid orientations), and (b--d)~three types of crossing gadgets used to ensure planarity.
  }
  \label{fig:cgs_auxiliary}
\end{figure}

\ifarxiv
  \begin{proof}
  We construct a reduction from \PlBiREMCC.
  As outlined in Section~\ref{sec:cgs}, we replace the vertices of the given instance with blue and red vertex gadgets according to their bipartition~(\cref{fig:cgs_vertex} in \cref{sec:cgs}).
  Edges are replaced by the corresponding edge gadgets~(\cref{fig:cgs_edge} in \cref{sec:cgs}).
  Hollow vertices of weight $1$ in the edge gadgets are simulated by \vAND, \vOR, and \vMAJ vertices as shown in \cref{fig:cgs_weight1}.
  To complete the proof, it suffices to show the correctness of these gadgets and construct the crossing gadgets.

  First, we show the validity of the blue vertex gadget.
  By construction, each edge gadget can orient at most one blue edge in the constraint graph outward.
  Summing over the three incident edge gadgets, at most three blue edges can be oriented toward a blue vertex gadget.
  Conversely, as a blue vertex gadget consists of three \textsc{or} vertices, it requires at least three incoming blue edges.
  Therefore, in each terminal, exactly one blue edge is oriented inward, satisfying the desired property.

  Next, we show the validity of the edge gadgets and the red vertex gadget.
  In an edge gadget, exactly one blue edge is oriented outward; thus, at least one of the three paths must be oriented from red to blue.
  From the perspective of the incident red vertex gadget, in each terminal, at most two red edges can be oriented inward.
  Summing over the three incident edge gadgets, at most six red edges can be oriented inward.
  Conversely, as a red vertex gadget consists of three \textsc{maj} vertices, it requires at least six incoming red edges.
  Therefore, similarly, in each terminal exactly two red edges are oriented from an edge gadget toward the red vertex gadget (and exactly one is oriented toward the edge gadget); this satisfies the desired property.

  Finally, we address planarity.
  The constructions of the crossing gadgets are illustrated in \cref{fig:cgs_crossing_red_red,fig:cgs_crossing_red_blue,fig:cgs_crossing_blue_blue}.
  In a red-red crossing gadget (\cref{fig:cgs_crossing_red_red}), no valid orientation exists if two opposite exterior edges are oriented outward.
  Conversely, if two adjacent exterior edges are oriented outward, the orientations of the interior edges are uniquely determined, forcing the remaining two exterior edges to be oriented inward.
  Hence, the orientations of the crossing red edges are consistent.
  The red-blue crossing gadget (\cref{fig:cgs_crossing_red_blue}) is constructed from two red-red crossing gadgets, while the blue-blue crossing gadget (\cref{fig:cgs_crossing_blue_blue}) consists of two red-blue crossing gadgets.
  It is straightforward to verify that these crossing gadgets function as intended.
\end{proof}

\fi

\subsection{\textsc{Kakuro}}\label{sec:kakuro}

\emph{Kakuro}~\cite{KakuroRule}, also known as \emph{Cross Sum}, is a well-known puzzle.
The rules of Kakuro are as follows:
\begin{enumerate}
  \item Given a grid, assign digits from $1$ to $9$ to each blank cell.
  \item For each horizontal or vertical run of contiguous blank cells, the sum of the digits in the run must equal the given clue.
  \item No digit may be repeated within a single run.
\end{enumerate}
\Cref{fig:kakuro_instance} shows an example of a Kakuro puzzle and its unique solution.
A decision problem \Kakuro asks whether a given Kakuro puzzle has a solution, and a counting problem \#\Kakuro asks for the number of solutions to a given Kakuro puzzle.
In standard \Kakuro, we consider only instances where every run contains at least two blank cells and that all blank cells are orthogonally connected.

\begin{figure}
  \begin{subfigure}[b]{0.48\linewidth}
    \centering
    \includegraphics{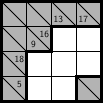}
    \caption{Instance}
  \end{subfigure}
  \hfill
  \begin{subfigure}[b]{0.48\linewidth}
    \centering
    \includegraphics{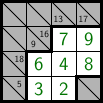}
    \caption{Solution}
  \end{subfigure}
  \caption{
    An example of a Kakuro puzzle and its solution.
    The numbers in the top-right corners of the black cells are clues for the horizontal runs, and those in the bottom-left corners are clues for the vertical runs.
  }
  \label{fig:kakuro_instance}
\end{figure}

We next consider a generalized version of \Kakuro introduced by Seta~\cite{Seta02CrossSum}.
For positive integers $N$, $\Lmin$, and $\Lmax$, \GenKakuro is a variant of \Kakuro with the following additional constraints:
\begin{itemize}
  \item The digits allowed in each cell are restricted to $\{1, \dots, N\}$.
  \item Each run of contiguous blank cells has a length of between $\Lmin$ and $\Lmax$.
\end{itemize}
By definition, we assume that $N \ge \Lmax \ge \Lmin$.
Under this definition, standard \Kakuro is equivalent to \GenKakuro[9][2][9].

For fixed $N$, any instance of \GenKakuro is trivially an instance of \GenKakuro[N][\Lmin'][\Lmax'] whenever $\Lmin' \le \Lmin \le \Lmax \le \Lmax'$, so hardness transfers to wider length intervals (whereas varying $N$ does not preserve solutions in this way).

The known complexity results for \Kakuro are summarized below:
\begin{itemize}
  \item \GenKakuro[N][2][2] is in P for all $N$~\cite{Seta02CrossSum}.
  \item \GenKakuro[N][2][5] is ASP-complete for all $N \ge 7$~\cite{Seta02CrossSum,YatoSeta03ASP}.
  \item \GenKakuro[N][1][3] is ASP-complete for all $N \ge 7$~\cite{Seta02CrossSum}.
  \item \GenKakuro[9][2][4] is ASP-complete~\cite{Ruepp10Kakuro}.
\end{itemize}

Our main result in this section is the following.

\begin{theorem}\label{thm:kakuro}
  \GenKakuro[N][2][3] is ASP-complete for all $N \ge 3$.
\end{theorem}

This result strengthens the previously known hardness results.
Moreover, any orthogonally connected puzzle board where all runs have length~$2$ must consist of a single $2 \times 2$ block, which can be solved in $O(1)$ time; thus \#\GenKakuro[N][2][2] is polynomial-time computable for all $N$.
Consequently, this theorem allows us to completely classify the counting complexity of \GenKakuro[N][2][\Lmax] for all $N \ge 2$ and $\Lmax \ge 2$.

\begin{corollary}
  \#\GenKakuro[N][2][\Lmax] is polynomial-time computable for $\Lmax = 2$, and \#P-complete and ASP-complete for $N \ge \Lmax \ge 3$.
\end{corollary}

\begin{proofof}{\cref{thm:kakuro}}
  We construct a reduction from \PlBiMCC.
  Let $G = (V_1 \cup V_2, E, A)$ be a given instance graph.
  The key idea is to encode the state of each edge in a cycle cover of $G$ using digits assigned to the terminals of Kakuro gadgets.
  Specifically, an edge directed from a vertex in $V_1$ to a vertex in $V_2$ corresponds to digit $2$, an edge in the opposite direction corresponds to digit $3$, and an unused edge corresponds to digit $1$.

  First, we construct vertex and edge gadgets.
  We employ two types of terminals of gadgets: \emph{red} (circled) and \emph{yellow} (boxed).
  A red terminal belongs to a horizontal run within the gadget and is exposed either upwards or downwards,
  whereas a yellow terminal belongs to a vertical run within the gadget and is exposed either leftwards or rightwards.
  Gadgets are connected by identifying a red terminal of one gadget with a yellow terminal of another, so that the shared cell forms a standard intersection of a horizontal and a vertical run without unwanted interference.
  Each vertex gadget has three terminals of the same color, determined by the bipartition of the vertex:
  red for $V_1$ and yellow for $V_2$.
  Each edge gadget connects terminals of opposite colors.

  Since each vertex in a cycle cover is incident to exactly one incoming edge, one outgoing edge, and one unused edge, the three terminals of each vertex gadget must be assigned a permutation of the distinct digits $\{1, 2, 3\}$.
  Furthermore, since every vertex in $G$ is incident to at least one directed edge, it suffices to construct two types of vertex gadgets per color.
  \Cref{fig:kakuro_vertex} illustrates the former.
  The lower-left terminal represents an incoming edge and admits digits in $\{1,3\}$ in the red vertex gadget (resp., $\{1,2\}$ in the yellow vertex gadget).
  The latter is obtained by replacing all clues $3$ with $4$ (resp., $4$ with $3$).
  Note that a $90^\circ$ rotation exchanges terminal colors, and thereby reverses the direction of the edge.

  \begin{figure}
    \begin{subfigure}{\linewidth}
      \begin{subfigure}{0.24\linewidth}
        \centering
        \includegraphics{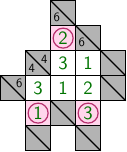}
      \end{subfigure}\hfill
      \begin{subfigure}{0.24\linewidth}
        \centering
        \includegraphics{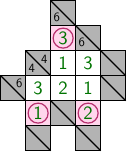}
      \end{subfigure}\hfill
      \begin{subfigure}{0.24\linewidth}
        \centering
        \includegraphics{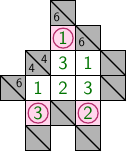}
      \end{subfigure}\hfill
      \begin{subfigure}{0.24\linewidth}
        \centering
        \includegraphics{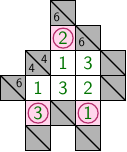}
      \end{subfigure}
      \caption{Red}\label{fig:kakuro_vertex_red}
    \end{subfigure}\vspace{1em}

    \begin{subfigure}{\linewidth}
      \begin{subfigure}{0.24\linewidth}
        \centering
        \includegraphics{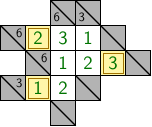}
      \end{subfigure}\hfill
      \begin{subfigure}{0.24\linewidth}
        \centering
        \includegraphics{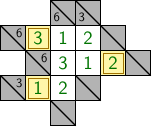}
      \end{subfigure}\hfill
      \begin{subfigure}{0.24\linewidth}
        \centering
        \includegraphics{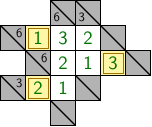}
      \end{subfigure}\hfill
      \begin{subfigure}{0.24\linewidth}
        \centering
        \includegraphics{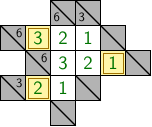}
      \end{subfigure}
      \caption{Yellow}\label{fig:kakuro_vertex_yellow}
    \end{subfigure}

    \caption{
      Vertex gadgets for \Kakuro (incident to a lower-left incoming edge) and their local solutions.
      In three terminals of each gadget, the assigned digits are different to one another, corresponding to the three incident edges of the vertex.
    }
    \label{fig:kakuro_vertex}
  \end{figure}

  Vertex gadgets are connected using three types of edge gadgets: directed (red-to-yellow and yellow-to-red) and undirected.
  \Cref{fig:kakuro_directed_edge} illustrates the components of a directed edge gadget from red to yellow.
  Each component admits exactly two local solutions, forcing the terminal values to be either both $1$ or both $2$.
  We can extend the gadget, bend it, or invert its terminal parity (i.e., shift the checkerboard parity of the exposed terminal cell by one).
  A directed edge gadget from yellow to red is obtained by swapping clue values $3$ and $4$.
  The components of an undirected edge gadget are illustrated in \cref{fig:kakuro_undirected_edge}.
  Each component admits exactly three local solutions, corresponding to terminal values $1$, $2$, or $3$.
  As with directed edges, these gadgets can be extended and bent.
  Note that no parity inversion component is needed for undirected edges, as in each vertex gadget, the two terminals that can admit undirected edges have the same parity.
  Therefore, we can configure the vertex gadgets so that all undirected edge gadgets connect terminals of matching parities.
  All these edge gadgets can, indeed, be placed so as not to interfere with one another or with vertex gadgets.

  \begin{figure}
    \begin{subfigure}{0.32\linewidth}
      \centering
      \includegraphics{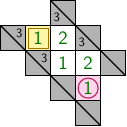}
      \caption{Straight}\label{fig:kakuro_directed_edge_straight}
    \end{subfigure}\hfill
    \begin{subfigure}{0.32\linewidth}
      \centering
      \includegraphics{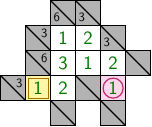}
      \caption{Turn}\label{fig:kakuro_directed_edge_turn}
    \end{subfigure}\hfill
    \begin{subfigure}{0.32\linewidth}
      \centering
      \includegraphics{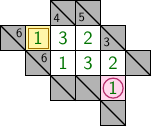}
      \caption{Parity inversion}\label{fig:kakuro_directed_edge_parity_inversion}
    \end{subfigure}
    \caption{
      Components of a directed edge gadget for \Kakuro (from red to yellow) and an example of their local solutions.
      The two terminals of each component admit the same digit: either $1$ or $2$.
      The straight component can be extended by inserting additional blank cells with appropriate clues.
    }
    \label{fig:kakuro_directed_edge}
  \end{figure}

  \begin{figure}
    \begin{subfigure}{0.48\linewidth}
      \centering
      \includegraphics{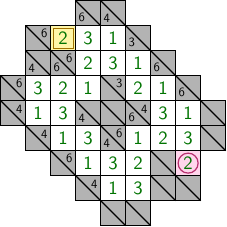}
      \caption{Straight}
    \end{subfigure}\hfill
    \begin{subfigure}{0.48\linewidth}
      \centering
      \includegraphics{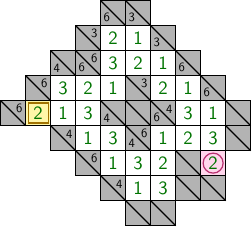}
      \caption{Turn}
    \end{subfigure}
    \caption{
      Components of an undirected edge gadget for \Kakuro and an example of their local solutions.
      The two terminals of each component admit the same digit: either $1$, $2$, or $3$.
      The straight component can be extended at the portion that has a similar structure to the straight component of a directed edge gadget.
    }
    \label{fig:kakuro_undirected_edge}
  \end{figure}

  The resulting \Kakuro instance contains runs of length only $2$ and $3$.
  Moreover, because all clues are sufficiently small (e.g., at most $4$ for length-$2$ runs and $6$ for length-$3$ runs), no digit greater than or equal to $4$ can participate in any valid assignment.
  Thus, the solution space is identical for any $N \ge 3$, preserving parsimony.
  This completes the proof.
\end{proofof}

\section{Orientation Model Equivalent to RE-MCC}\label{sec:orientation_model}

In this section we present an orientation model equivalent to \PlBiREMCC that facilitates reduction via rectilinear embeddings.
We begin with the definition of a half-integral orientation and its validity under vertex weights.

\begin{definition}
  Let $G = (V, E)$ be an undirected graph, and let $\vec{E} = \{(u, v), (v, u) \mid uv \in E\}$.
  A \emph{half-integral orientation} (or simply, \emph{orientation}) of $G$ is a function $\alpha \colon \vec{E} \to \{0, 1/2, 1\}$ such that
  \[
    \alpha(u, v) + \alpha(v, u) = 1
  \]
  for all $uv \in E$.
  For a vertex-weighted graph $(G, w)$ with $w \colon V \to \RR_{\ge 0}$, an orientation $\alpha$ is \emph{valid} if
  \[
    \sum_{u \in N(v)} \alpha(u, v) = w(v)
  \]
  for all $v \in V$, where $N(v)$ denotes the set of neighbors of $v$.
  This condition is referred to as the \emph{in-degree condition}.
\end{definition}

In our model, each edge $uv$ is assigned an \emph{edge type} that specifies the set of its admissible orientations.
We employ four types of edges:
\begin{description}
  \item[Free] $\alpha(u, v) \in \{0, 1/2, 1\}$.
        Admits any half-integral orientation.
  \item[Biased] $\alpha(u, v) \in \{1/2, 1\}$.
        Must direct at least $1/2$ toward $v$ (biased toward $v$).
  \item[Integral] $\alpha(u, v) \in \{0, 1\}$.
        Must be fully directed toward either $u$ or $v$.
  \item[Neutral] $\alpha(u, v) = 1/2$.
        Contributes equally to both endpoints.
\end{description}

Formally, we transform an instance of \PlBiREMCC into an orientation instance as follows.
First, we rectilinearly embed a given \PlBiREMCC instance into a grid (\cref{fig:rectilinear_embedding}).
Then, we subdivide each edge of the rectilinear embedding into length-$1$ segments, and construct a rectangular grid graph $H$.
Finally, we assign vertex weights and edge constraints of $H$ as follows (\cref{fig:orientation_model}):
(i) Each vertex $p \in V_H$ is assigned the target weight $w(p) = \frac{1}{2} \deg_H(p)$, and
(ii) each edge $e \in E_H$ inherits its constraint from the corresponding segment:
an edge is free if the segment corresponds to an undirected edge in $G$;
biased (oriented consistently) if it corresponds to a directed edge;
integral if it corresponds to a required edge;
and neutral if it is an empty segment.

\begin{figure}
  \begin{subfigure}{0.48\textwidth}
    \centering
    \includegraphics{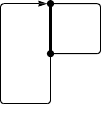}
    \caption{Rectilinear embedding of $G$}\label{fig:rectilinear_embedding}
  \end{subfigure}\hfill
  \begin{subfigure}{0.48\textwidth}
    \centering
    \includegraphics{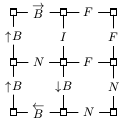}
    \caption{Transformed grid graph $H$}\label{fig:orientation_model}
  \end{subfigure}
  \caption{
    An example of the orientation model construction.
    Labels on the edges of $H$ indicate their types: $F$ for free, $B$ for biased (toward the arrow), $I$ for integral, and $N$ for neutral.
  }
  \label{fig:orientation_model_example}
\end{figure}

Now, we associate cycle covers of $G$ with valid orientations of $H$ as follows.
For each vertex $p \in V_H$ and an incident edge $p q$, the edge is incoming to $p$, outgoing from $p$, or unused in the cycle cover if and only if $\alpha(p, q)$ is $1$, $0$, or $1/2$, respectively.
\ifarxiv
  Under this correspondence, we have the following lemma.
\else
  Under this correspondence, we have the following lemma (see the full version for the proof).
\fi

\begin{restatable}{lemma}{lemorientationmodel}
  \label{lem:orientation_model_equivalence}
  Let $G$ be an instance of \PlBiREMCC, and let $(H, w)$ be a grid graph of the orientation model constructed from $G$.
  Then there is a bijection between the cycle covers of $G$ and the valid orientations of $H$.
  Moreover, the bijection holds even if the in-degree condition is relaxed to inequalities (i.e., $\sum \alpha(q, p) \ge w(p)$ for all $p \in V_H$ or $\le w(p)$ for all $p \in V_H$).
\end{restatable}

\ifarxiv
  \begin{proof}
  We first verify the forward direction.
  Fix a valid cycle cover of $G$ and construct an orientation $\alpha$ of $H$ as follows:
  For each $pq \in E_H$, if it lies on the path corresponding to an edge $uv \in E$ and the cycle cover traverses $uv$ in the direction from $p$ toward $q$, set $\alpha(p, q) = 1$ and $\alpha(q, p) = 0$;
  if it traverses in the opposite direction, set $\alpha(p, q) = 0$ and $\alpha(q, p) = 1$;
  otherwise (if the edge is unused, or if $pq$ corresponds to an empty segment), set $\alpha(p, q) = \alpha(q, p) = 1/2$.
  By construction, $\alpha$ respects all edge type constraints and satisfies the in-degree condition at every vertex of $H$.

  Conversely, we show that any valid orientation $\alpha$ of $H$ uniquely induces a cycle cover of $G$ via the inverse of the forward construction.
  Specifically, we construct an oriented subgraph of $G$ where each edge $uv \in E$ is directed from $u$ to $v$ if the corresponding edges in $H$ are oriented with value $1$ toward $v$, and is omitted if their value is $1/2$.
  First, this correspondence is well-defined:
  at each intermediate vertex $p$ in $H$ along the path corresponding to an edge in $G$, the in-degree condition $\sum_{q \in N(p)} \alpha(q, p) = \frac{1}{2} \deg_H(p)$ ensures that the orientation value is propagated consistently along the entire path.
  Moreover, the constraints on biased edges ensure that the orientation respects the directions of the directed edges in $G$.

  It remains to show that every vertex in the constructed subgraph has in-degree and out-degree exactly $1$.
  Let $v \in V_G$ be a vertex in $G$, and let $p_v \in V_H$ be its corresponding vertex.
  The vertex $p_v$ is incident to three non-neutral edges corresponding to the edges in $G$ incident to $v$, plus at most one neutral edge.
  Since each neutral edge contributes exactly $1/2$ to the in-degree of $p_v$, subtracting its contribution from the condition $\sum \alpha(u, p_v) = \frac{1}{2}\deg_H(p_v)$ implies that the sum of the orientation values from the three incident edges is always exactly $3/2$.
  Recalling that each vertex in $G$ is incident to exactly one required edge, among these three edges, exactly one is integral (contributing either $0$ or $1$).
  Since each edge contributes a value from $\{0, 1/2, 1\}$:
  \begin{itemize}
    \item if the integral edge contributes $1$, the remaining two edges must sum to $1/2$, forcing their values to be $\{0, 1/2\}$;
    \item if the integral edge contributes $0$, the remaining two edges must sum to $3/2$, forcing their values to be $\{1/2, 1\}$.
  \end{itemize}
  In either case, the multiset of orientation values directed toward $p_v$ is uniquely determined to be $\{0, 1/2, 1\}$.
  This means that $v$ has exactly one incoming edge (value $1$), one outgoing edge (value $0$), and one unused edge (value $1/2$), confirming that the subgraph is a valid cycle cover.

  Finally, even if the in-degree condition is relaxed to inequalities, any valid orientation must satisfy it with equality, and thus the above argument still holds.
  Indeed, summing the condition over all vertices yields $|E_H|$ on both sides by the handshaking lemma, forcing equality at every vertex.
\end{proof}

\fi

Via this orientation model, we reduce \PlBiREMCC to various pencil-and-paper puzzles in a parsimonious manner.
Specifically, we establish new ASP-completeness results for the following puzzles: \textsc{Chocona}, \textsc{Four Cells}, \textsc{Hinge}, and \textsc{Shimaguni}.
Moreover, we strengthen the known NP-completeness results for \textsc{Choco Banana}~\cite{Iwamoto2024ChocoBanana} and \textsc{Five Cells}~\cite{Iwamoto2022FiveCells} to ASP-completeness.

We next describe the key intuition underlying these reductions.
In what follows, to clear denominators and facilitate reductions to discrete grid puzzles, we scale all values by a factor of $2$.
Under this equivalent integer formulation, each edge takes a value in $\{0, 1, 2\}$, and the in-degree condition becomes $\sum_{u \in N(v)} \alpha(u, v) = \deg(v)$.

For each puzzle, the reduction constructs vertex gadgets and edge gadgets.
These gadgets simulate the behavior of vertices and edges in the orientation model as follows:
At each interface between an incident vertex gadget and edge gadget, they share a capacity of two units.
Each vertex gadget is designed to receive a total of $\deg(v)$ units (or, depending on the puzzle's regional rules, at least or at most $\deg(v)$ units) across all its incident interfaces.
Each edge gadget distributes a total of two units between its two endpoints, with its internal states restricted according to the edge type (integral, biased, or neutral).

As a representative example of this framework, \cref{fig:shimaguni_gadgets} illustrates the reduction from the orientation model to \textsc{Shimaguni}.
\ifarxiv\else
Summarizing our applications of this framework, we obtain the following theorem (see the full version for the complete reductions and correctness proofs).

\begin{theorem}\label{thm:puzzles}
  \textsc{Choco Banana}, \textsc{Chocona}, \textsc{Five Cells}, \textsc{Four Cells}, \textsc{Hinge}, and \textsc{Shimaguni} are ASP-complete.
\end{theorem}
\fi

\begin{figure}
  \centering
  \includegraphics[scale=\pdfscale]{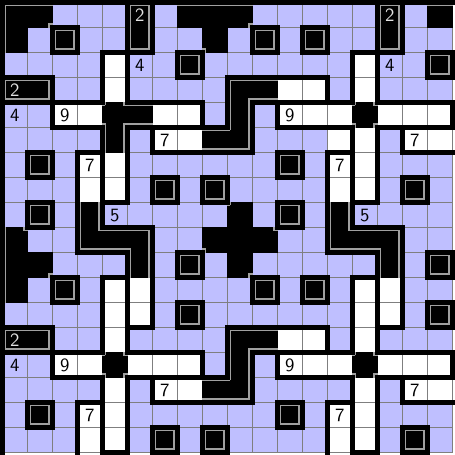}
  \caption{
    An example of the reduction from the orientation model to \textsc{Shimaguni}.
    The regions with clue value $9$ correspond to the vertex gadgets, while the regions with clue value $7$ correspond to the edge gadgets.
    A vertex gadget and an incident edge gadget share a border along two pairs of adjacent cells; by the rules of \textsc{Shimaguni}, no adjacent pair across the region border can be simultaneously shaded.
  }
  \label{fig:shimaguni_gadgets}
\end{figure}

\section{\textsc{Application to Puzzles}}\label{sec:puzzles}

In this section, we describe how the orientation model from \cref{sec:orientation_model} can be applied to various pencil-and-paper puzzles.

\subsection{Shimaguni}

\emph{Shimaguni}~\cite{ShimaguniRule} (lit. island country) is a puzzle first introduced in \textit{Puzzle Communication Nikoli} Vol.~117.
Given a rectangular grid partitioned into regions, the goal is to shade some cells according to the following rules:
Each region must contain exactly one orthogonally connected block of shaded cells (called an ``island'').
A number in a region indicates the area of the island in that region.
Islands in different regions cannot be orthogonally adjacent.
Finally, islands in orthogonally adjacent regions must have different sizes.
\Cref{fig:shimaguni_instance} shows an instance of Shimaguni and its unique solution.
A decision problem \textsc{Shimaguni} asks whether a given Shimaguni puzzle has a solution.

\begin{figure}[tbp]
  \centering
  \begin{subfigure}{0.48\linewidth}
    \centering
    \includegraphics[scale=\pdfscale]{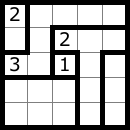}
    \caption{Instance}
  \end{subfigure}
  \hfill
  \begin{subfigure}{0.48\linewidth}
    \centering
    \includegraphics[scale=\pdfscale]{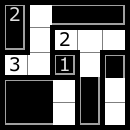}
    \caption{Solution}
  \end{subfigure}
  \caption{An instance of Shimaguni and its solution.}
  \label{fig:shimaguni_instance}
\end{figure}

The whole reduction framework is illustrated in \cref{fig:shimaguni_gadgets} in \cref{sec:orientation_model}.
By the rule that islands in different regions cannot be orthogonally adjacent, these gadgets function as intended.
\Cref{fig:shimaguni_edge_others} illustrates three variants of edge gadgets, each of which restricts the admissible orientations as required.

\begin{figure}[tbp]
  \centering
  \begin{subfigure}{0.32\linewidth}
    \centering
    \includegraphics[scale=\pdfscale]{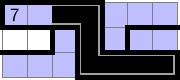}
    \caption{Oriented with weights $(0, 2)$}\label{fig:shimaguni_edge_02}
  \end{subfigure}
  \hfill
  \begin{subfigure}{0.32\linewidth}
    \centering
    \includegraphics[scale=\pdfscale]{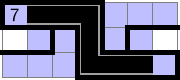}
    \caption{Oriented with weights $(1, 1)$}\label{fig:shimaguni_edge_11}
  \end{subfigure}
  \hfill
  \begin{subfigure}{0.32\linewidth}
    \centering
    \includegraphics[scale=\pdfscale]{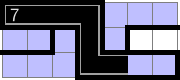}
    \caption{Oriented with weights $(2, 0)$}\label{fig:shimaguni_edge_20}
  \end{subfigure}
  \caption{
    An edge gadget for Shimaguni and its three local solutions.
    In each solution, the number of unshaded cells adjacent to each of the two incident vertex gadgets corresponds to its contribution to the in-degree of that vertex.
  }
  \label{fig:shimaguni_edge}
\end{figure}

\begin{figure}[tbp]
  \centering
  \begin{subfigure}{0.32\linewidth}
    \centering
    \includegraphics[scale=\pdfscale]{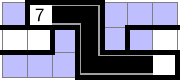}
    \caption{Biased (rightward)}\label{fig:shimaguni_edge_biased}
  \end{subfigure}
  \hfill
  \begin{subfigure}{0.32\linewidth}
    \centering
    \includegraphics[scale=\pdfscale]{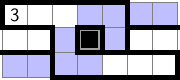}
    \caption{Integral}\label{fig:shimaguni_edge_integral}
  \end{subfigure}
  \hfill
  \begin{subfigure}{0.32\linewidth}
    \centering
    \includegraphics[scale=\pdfscale]{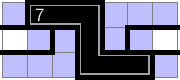}
    \caption{Neutral}\label{fig:shimaguni_edge_neutral}
  \end{subfigure}
  \caption{
    Variants of the edge gadgets for Shimaguni.
  }
  \label{fig:shimaguni_edge_others}
\end{figure}

Because this reduction is parsimonious, we obtain the following result.

\begin{theorem}
  Shimaguni is ASP-complete.
\end{theorem}

\subsection{Choco Banana}

\emph{Choco Banana}~\cite{ChocoBananaRule} is a puzzle first introduced in \textit{Puzzle Communication Nikoli} Vol.~176.
Given a rectangular grid, the goal is to shade some cells such that every orthogonally connected block of shaded cells forms a rectangle (or a square), while no orthogonally connected block of unshaded cells forms a rectangle.
A number indicates the area of the block (shaded or unshaded) containing the cell.
\Cref{fig:chocobanana_instance} shows an instance of Choco Banana and its unique solution.
A decision problem \textsc{Choco Banana} asks whether a given Choco Banana puzzle has a solution.
\textsc{Choco Banana} is known to be NP-complete~\cite{Iwamoto2024ChocoBanana}.

\begin{figure}[tbp]
  \centering
  \begin{subfigure}{0.48\linewidth}
    \centering
    \includegraphics[scale=\pdfscale]{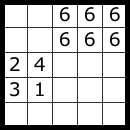}
    \caption{Instance}
  \end{subfigure}
  \hfill
  \begin{subfigure}{0.48\linewidth}
    \centering
    \includegraphics[scale=\pdfscale]{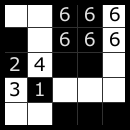}
    \caption{Solution}
  \end{subfigure}
  \caption{An instance of Choco Banana and its solution.}
  \label{fig:chocobanana_instance}
\end{figure}

Before constructing the reduction, we note an observation about Choco Banana:
If three cells containing the same prime number $p$ form an L-shaped block, they must all be unshaded.
This is verified by a careful case analysis of all shading assignments, given that a shaded rectangle of area $p$ is restricted to a $1 \times p$ or $p \times 1$ shape, and no unshaded block with a clue $p$ can form such a shape.
(This observation is a known technique among avid Choco Banana solvers.)

Now we construct a reduction via the orientation model as illustrated in \cref{fig:chocobanana_gadgets}.
The configuration of unshaded cells with clues $37$ and $13$ is fixed by the observation above.
Each block of unshaded cells containing the clue $37$ corresponds to a vertex gadget, while each block of undetermined cells corresponds to an edge gadget.
These gadgets correctly realize the model, as shown below:
Each vertex gadget requires unshaded cells equal to the difference between its current size and $37$, which equals its degree.
On the other hand, each edge gadget contributes a total of two unshaded cells to the adjacent vertex gadgets (\cref{fig:chocobanana_edge}).
\Cref{fig:chocobanana_edge_others} illustrates three variants of edge gadgets, each of which restricts the admissible orientations as required.

\begin{figure}[tbp]
  \centering
  \includegraphics[scale=\pdfscale]{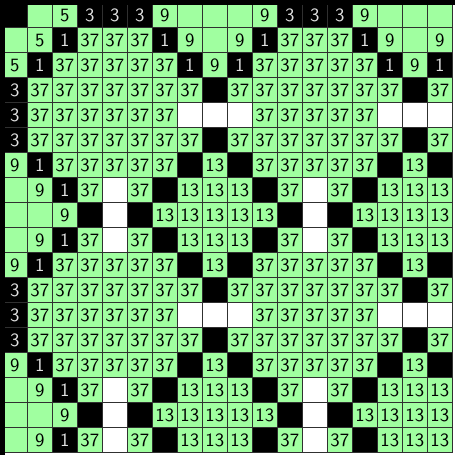}
  \caption{
    Global arrangement of edge and vertex gadgets for Choco Banana.
    Degree-$2$ vertex at upper left, degree-$3$ vertices at upper right and lower left, and degree-$4$ vertex at lower right.
  }
  \label{fig:chocobanana_gadgets}
\end{figure}

\begin{figure}[tbp]
  \centering
  \begin{subfigure}{0.32\linewidth}
    \centering
    \includegraphics[scale=\pdfscale]{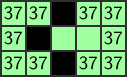}
    \caption{Oriented with weights $(0, 2)$}\label{fig:chocobanana_edge_02}
  \end{subfigure}
  \hfill
  \begin{subfigure}{0.32\linewidth}
    \centering
    \includegraphics[scale=\pdfscale]{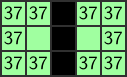}
    \caption{Oriented with weights $(1, 1)$}\label{fig:chocobanana_edge_11}
  \end{subfigure}
  \hfill
  \begin{subfigure}{0.32\linewidth}
    \centering
    \includegraphics[scale=\pdfscale]{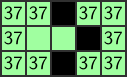}
    \caption{Oriented with weights $(2, 0)$}\label{fig:chocobanana_edge_20}
  \end{subfigure}
  \caption{
    An edge gadget for Choco Banana and its three local solutions, corresponding to the admissible orientations.
  }
  \label{fig:chocobanana_edge}
\end{figure}

\begin{figure}[tbp]
  \centering
  \begin{subfigure}{0.32\linewidth}
    \centering
    \includegraphics[scale=\pdfscale]{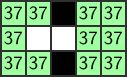}
    \caption{Biased (rightward)}\label{fig:chocobanana_edge_biased}
  \end{subfigure}
  \hfill
  \begin{subfigure}{0.32\linewidth}
    \centering
    \includegraphics[scale=\pdfscale]{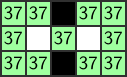}
    \caption{Integral}\label{fig:chocobanana_edge_integral}
  \end{subfigure}
  \hfill
  \begin{subfigure}{0.32\linewidth}
    \centering
    \includegraphics[scale=\pdfscale]{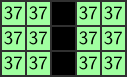}
    \caption{Neutral}\label{fig:chocobanana_edge_neutral}
  \end{subfigure}
  \caption{
    Variants of the edge gadgets for Choco Banana.
  }
  \label{fig:chocobanana_edge_others}
\end{figure}

Because this reduction is parsimonious, we obtain the following result.

\begin{theorem}
  Choco Banana is ASP-complete.
\end{theorem}

\subsection{Hinge}

\emph{Hinge}~\cite{HingeRule} is a puzzle first introduced in \textit{Puzzle Communication Nikoli} Vol.~159.
Given a rectangular grid partitioned into regions, the goal is to shade cells according to the following rules:
Each orthogonally connected block of shaded cells must be intersected exactly once by a single straight line segment of region boundaries (called a ``hinge'').
Furthermore, each such block must be symmetric with respect to its hinge.
A number in a region indicates the number of shaded cells contained in that region.
\Cref{fig:hinge_instance} shows an instance of Hinge and its unique solution.
A decision problem \textsc{Hinge} asks whether a given Hinge puzzle has a solution.

\begin{figure}[tbp]
  \centering
  \begin{subfigure}{0.48\linewidth}
    \centering
    \includegraphics[scale=\pdfscale]{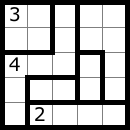}
    \caption{Instance}
  \end{subfigure}
  \hfill
  \begin{subfigure}{0.48\linewidth}
    \centering
    \includegraphics[scale=\pdfscale]{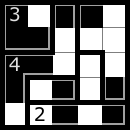}
    \caption{Solution}
  \end{subfigure}
  \caption{An instance of Hinge and its solution.}
  \label{fig:hinge_instance}
\end{figure}

We construct a reduction via the orientation model as illustrated in \cref{fig:hinge_flow_model}.
Each region with arms corresponds to a vertex gadget with degree equal to the number of arms, while each $1 \times 3$ or $3 \times 1$ region that connects two arms corresponds to an edge gadget.
These gadgets correctly realize the model, as shown below:
Each edge gadget supplies a total of two shaded cells to the two connected vertex gadgets (\cref{fig:hinge_edge}).
On the other hand, each vertex gadget requires at least as many shaded cells as its degree from its adjacent edge gadgets.
By \cref{lem:orientation_model_equivalence}, relaxing the in-degree condition to an inequality suffices for the reduction.
\Cref{fig:hinge_edge_others} illustrates three variants of edge gadgets, each of which restricts the admissible orientations as required.

\begin{figure}[tbp]
  \centering
  \includegraphics[scale=\pdfscale]{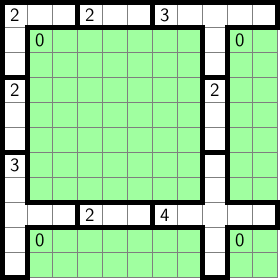}
  \caption{
    Global arrangement of edge and vertex gadgets for Hinge.
    Degree-$2$ vertex at upper left, degree-$3$ vertices at upper right and lower left, and degree-$4$ vertex at lower right.
  }
  \label{fig:hinge_flow_model}
\end{figure}

\begin{figure}[tbp]
  \centering
  \begin{subfigure}{0.32\linewidth}
    \centering
    \includegraphics[scale=\pdfscale]{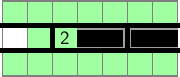}
    \caption{Oriented with weights $(0, 2)$}\label{fig:hinge_edge_02}
  \end{subfigure}
  \hfill
  \begin{subfigure}{0.32\linewidth}
    \centering
    \includegraphics[scale=\pdfscale]{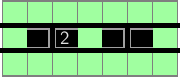}
    \caption{Oriented with weights $(1, 1)$}\label{fig:hinge_edge_11}
  \end{subfigure}
  \hfill
  \begin{subfigure}{0.32\linewidth}
    \centering
    \includegraphics[scale=\pdfscale]{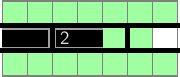}
    \caption{Oriented with weights $(2, 0)$}\label{fig:hinge_edge_20}
  \end{subfigure}
  \caption{
    An edge gadget for Hinge and its three local solutions, corresponding to the admissible orientations.
  }
  \label{fig:hinge_edge}
\end{figure}

\begin{figure}[tbp]
  \centering
  \begin{subfigure}{0.32\linewidth}
    \centering
    \includegraphics[scale=\pdfscale]{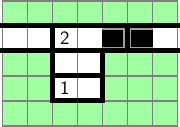}
    \caption{Biased (rightward)}\label{fig:hinge_edge_biased}
  \end{subfigure}
  \hfill
  \begin{subfigure}{0.32\linewidth}
    \centering
    \includegraphics[scale=\pdfscale]{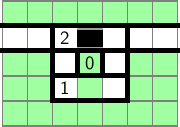}
    \caption{Integral}\label{fig:hinge_edge_integral}
  \end{subfigure}
  \hfill
  \begin{subfigure}{0.32\linewidth}
    \centering
    \includegraphics[scale=\pdfscale]{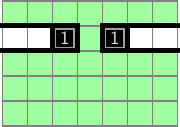}
    \caption{Neutral}\label{fig:hinge_edge_neutral}
  \end{subfigure}
  \caption{
    Variants of the edge gadgets for Hinge.
  }
  \label{fig:hinge_edge_others}
\end{figure}

These edge gadgets can be arranged so as not to interfere with one another.
Because this reduction is parsimonious, we obtain the following result.

\begin{theorem}
  Hinge is ASP-complete.
\end{theorem}

\subsection{Chocona}

\emph{Chocona}~\cite{ChoconaRule} is a puzzle first introduced in \textit{Puzzle Communication Nikoli} Vol.~122.
Given a rectangular grid partitioned into regions, the goal is to shade some cells so that every orthogonally connected block of shaded cells forms a rectangle (or a square).
Each number indicates the number of shaded cells contained in that region.
\Cref{fig:chocona_instance} shows an instance of Chocona and its unique solution.
A decision problem \textsc{Chocona} asks whether a given Chocona puzzle has a solution.

\begin{figure}[tbp]
  \centering
  \begin{subfigure}{0.48\linewidth}
    \centering
    \includegraphics[scale=\pdfscale]{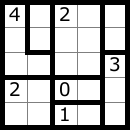}
    \caption{Instance}
  \end{subfigure}
  \hfill
  \begin{subfigure}{0.48\linewidth}
    \centering
    \includegraphics[scale=\pdfscale]{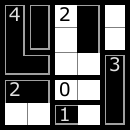}
    \caption{Solution}
  \end{subfigure}
  \caption{An instance of Chocona and its solution.}
  \label{fig:chocona_instance}
\end{figure}

We construct a reduction via the orientation model as illustrated in \cref{fig:chocona_flow_model}.
Each region with four arms corresponds to a vertex gadget, while each block of undetermined cells that connects two arms of vertex gadgets corresponds to an edge gadget.
These gadgets correctly realize the model, as shown below:
Each edge gadget supplies a total of two shaded cells to the two connected vertex gadgets (\cref{fig:chocona_edge}).
On the other hand, each vertex gadget contains a clue value equal to its degree plus one and exactly one extra undetermined cell other than the connections to edge gadgets;
thus, each vertex gadget requires at least as many shaded cells as its degree from its adjacent edge gadgets.
By \cref{lem:orientation_model_equivalence}, relaxing the in-degree condition to an inequality suffices for the reduction.
\Cref{fig:chocona_edge_others} illustrates three variants of edge gadgets, each of which restricts the admissible orientations as required.

\begin{figure}[tbp]
  \centering
  \includegraphics[scale=\pdfscale]{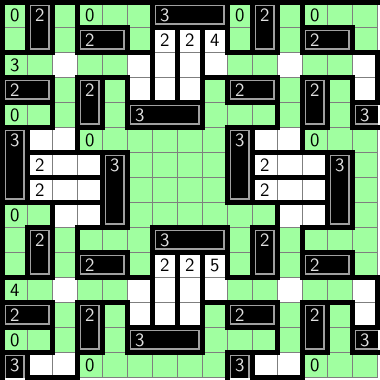}
  \caption{
    Global arrangement of edge and vertex gadgets for Chocona.
    Degree-$2$ vertex at upper left, degree-$3$ vertices at upper right and lower left, and degree-$4$ vertex at lower right.
  }
  \label{fig:chocona_flow_model}
\end{figure}

\begin{figure}[tbp]
  \centering
  \begin{subfigure}{0.32\linewidth}
    \centering
    \includegraphics[scale=\pdfscale]{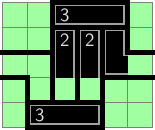}
    \caption{Oriented with weights $(0, 2)$}\label{fig:chocona_edge_02}
  \end{subfigure}
  \hfill
  \begin{subfigure}{0.32\linewidth}
    \centering
    \includegraphics[scale=\pdfscale]{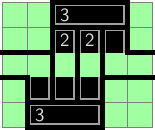}
    \caption{Oriented with weights $(1, 1)$}\label{fig:chocona_edge_11}
  \end{subfigure}
  \hfill
  \begin{subfigure}{0.32\linewidth}
    \centering
    \includegraphics[scale=\pdfscale]{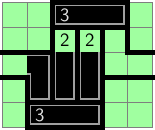}
    \caption{Oriented with weights $(2, 0)$}\label{fig:chocona_edge_20}
  \end{subfigure}
  \caption{
    An edge gadget for Chocona and its three local solutions, corresponding to the admissible orientations.
  }
  \label{fig:chocona_edge}
\end{figure}

\begin{figure}[tbp]
  \centering
  \begin{subfigure}{0.32\linewidth}
    \centering
    \includegraphics[scale=\pdfscale]{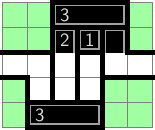}
    \caption{Biased (rightward)}\label{fig:chocona_edge_biased}
  \end{subfigure}
  \hfill
  \begin{subfigure}{0.32\linewidth}
    \centering
    \includegraphics[scale=\pdfscale]{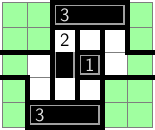}
    \caption{Integral}\label{fig:chocona_edge_integral}
  \end{subfigure}
  \hfill
  \begin{subfigure}{0.32\linewidth}
    \centering
    \includegraphics[scale=\pdfscale]{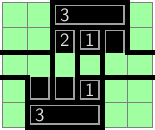}
    \caption{Neutral}\label{fig:chocona_edge_neutral}
  \end{subfigure}
  \caption{
    Variants of the edge gadgets for Chocona.
  }
  \label{fig:chocona_edge_others}
\end{figure}

Because this reduction is parsimonious, we obtain the following result.
\begin{theorem}
  Chocona is ASP-complete.
\end{theorem}

\subsection{Five Cells}\label{sec:five_cells}

\emph{Five Cells}~\cite{FiveCellsRule} is a puzzle first introduced in \textit{Puzzle Communication Nikoli} Vol.~133.
The goal is to partition a board into pentominoes.
A number in a cell indicates how many edges of that cell belong to the pentomino boundary.
\Cref{fig:fivecells_instance} shows an instance of Five Cells and its unique solution.
A decision problem \textsc{Five Cells} asks whether a given Five Cells puzzle has a solution.
\textsc{Five Cells} is known to be NP-complete~\cite{Iwamoto2022FiveCells}.

\begin{figure}[tbp]
  \centering
  \begin{subfigure}{0.48\linewidth}
    \centering
    \includegraphics[scale=\pdfscale]{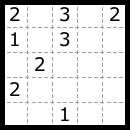}
    \caption{Instance}
  \end{subfigure}
  \hfill
  \begin{subfigure}{0.48\linewidth}
    \centering
    \includegraphics[scale=\pdfscale]{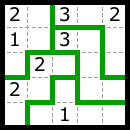}
    \caption{Solution}
  \end{subfigure}
  \caption{An instance of Five Cells and its solution.}
  \label{fig:fivecells_instance}
\end{figure}

We construct a reduction via the orientation model.
\Cref{fig:fivecells_gadget} illustrates a unit gadget.
The configuration of the surrounding yellow region is uniquely determined by the boundary conditions, specifically where the gadget abuts the board border or adjacent gadgets.
The central cell of the gadget corresponds to a vertex, from which four wires emanate in the four cardinal directions.
Conversely, the interface between two wires of adjacent unit gadgets corresponds to an edge.
These gadgets correctly realize the model, as shown below:
The central cell requires a total of $4$ cells from the four wires to form a pentomino, while each wire supplies a specific number of cells to the central cell (\cref{fig:fivecells_wire}).
Specifically, a wire connected to the board border supplies exactly $1$ cell, effectively adjusting the remaining demand of the central cell to $3$ (when abutting one border) or $2$ (when abutting two borders).
For the interface between two wires, the sum of their supplies is constrained to $2$; collectively, these two connected wires function as an edge gadget.

\begin{figure}[tbp]
  \centering
  \includegraphics[scale=\pdfscale]{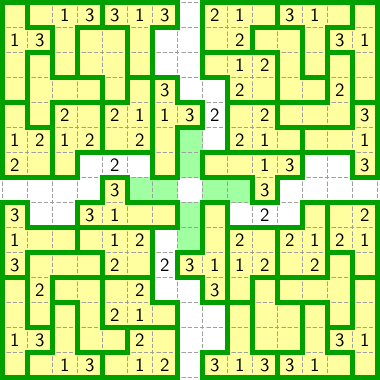}
  \caption{
    Unit gadget for Five Cells.
    The undetermined portion is divided into the central cell and four wires.
    The green cells indicate the connections between the central cell and the wires.
  }
  \label{fig:fivecells_gadget}
\end{figure}

\begin{figure}[tbp]
  \centering
  \begin{subfigure}{0.32\linewidth}
    \centering
    \includegraphics[scale=\pdfscale]{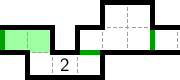}
    \caption{Weight $0$}\label{fig:fivecells_wire_0}
  \end{subfigure}
  \hfill
  \begin{subfigure}{0.32\linewidth}
    \centering
    \includegraphics[scale=\pdfscale]{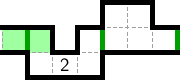}
    \caption{Weight $1$}\label{fig:fivecells_wire_1}
  \end{subfigure}
  \hfill
  \begin{subfigure}{0.32\linewidth}
    \centering
    \includegraphics[scale=\pdfscale]{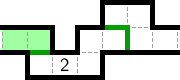}
    \caption{Weight $2$}\label{fig:fivecells_wire_2}
  \end{subfigure}
  \caption{
    A wire of the unit gadget for Five Cells and its three local solutions.
    The number of green cells connected to the central cell corresponds to the weight contributed to that vertex.
    The end opposite to the central cell connects to the board border or to an end of another wire.
  }
  \label{fig:fivecells_wire}
\end{figure}

To simulate the variants of edges, we introduce two specific types of wires (\cref{fig:fivecells_wire_type}).
An \emph{incoming} wire (\cref{fig:fivecells_wire_incoming}) can supply either $1$ or $2$ cells, whereas an \emph{exclusive} wire (\cref{fig:fivecells_wire_exclusive}) is restricted to supplying either $0$ or $2$ cells.
An edge gadget with exactly one incoming wire corresponds to a biased edge oriented toward the incoming side.
An edge gadget with two incoming wires corresponds to a neutral edge, as the only pair of weights summing to $2$ is $(1,1)$.
An edge gadget composed of exclusive wires corresponds to an integral edge.

\begin{figure}[tbp]
  \centering
  \begin{subfigure}{0.48\textwidth}
    \centering
    \includegraphics[scale=\pdfscale]{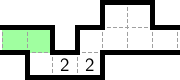}
    \caption{Incoming}
    \label{fig:fivecells_wire_incoming}
  \end{subfigure}\hfill
  \begin{subfigure}{0.48\textwidth}
    \centering
    \includegraphics[scale=\pdfscale]{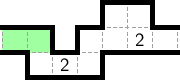}
    \caption{Exclusive}
    \label{fig:fivecells_wire_exclusive}
  \end{subfigure}
  \caption{
    Two specific types of wires for Five Cells, each restricting the local solutions and thereby the admissible orientations.
  }
  \label{fig:fivecells_wire_type}
\end{figure}

Because this reduction is parsimonious, we obtain the following result.

\begin{theorem}
  Five Cells is ASP-complete.
\end{theorem}

\subsection{Four Cells}

\emph{Four Cells}~\cite{FourCellsRule} is a puzzle first introduced in \textit{Puzzle Communication Nikoli} Vol.~132.
The goal is to partition a board into tetrominoes.
A number in a cell indicates how many edges of that cell belong to the tetromino boundary.
\Cref{fig:fourcells_instance} shows an instance of Four Cells and its unique solution.
A decision problem \textsc{Four Cells} asks whether a given Four Cells puzzle has a solution.

\begin{figure}[tbp]
  \centering
  \begin{subfigure}{0.48\linewidth}
    \centering
    \includegraphics[scale=\pdfscale]{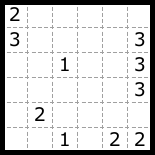}
    \caption{Instance}
  \end{subfigure}
  \hfill
  \begin{subfigure}{0.48\linewidth}
    \centering
    \includegraphics[scale=\pdfscale]{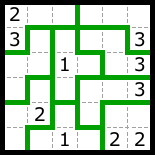}
    \caption{Solution}
  \end{subfigure}
  \caption{An instance of Four Cells and its solution.}
  \label{fig:fourcells_instance}
\end{figure}

We construct a reduction via the orientation model.
To construct the reduction, we employ the following three types of the orientation model vertices:
\begin{description}
  \item[T-Shaped] Degree at least $3$.
        Incident to two edges of the opposite direction, each of which can be free or biased, and one integral edge.
        The remaining direction is, if any, occupied by neutral edges.
  \item[L-Shaped] Degree at least $2$.
        Incident to a specified pair of adjacent non-neutral edges: two free, two integral, or two biased (one inward and one outward).
        Any remaining directions are occupied by neutral edges.
  \item[Neutral] Incident only to neutral edges.
\end{description}
The T-shaped vertices correspond to vertices in the original graph, while the three types of L-shaped vertices (free, integral, biased) correspond to undirected, required, and directed edges, respectively (including their orientation).
These vertices are sufficient to represent any \PlBiREMCC instances.
Indeed, $2 \times 2$ units of T-shaped, L-shaped, and neutral vertices as shown in \cref{fig:zigzag_embedding} and their rotations can be arranged to form a embedding of any \PlBiREMCC instance.

\begin{figure}[tbp]
  \begin{subfigure}{0.32\linewidth}
    \centering
    \includegraphics{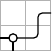}
    \hspace{1em}
    \includegraphics{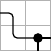}
    \caption{Vertices}
  \end{subfigure}
  \hfill
  \begin{subfigure}{0.32\linewidth}
    \centering
    \includegraphics{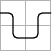}
    \caption{Straight segment}
  \end{subfigure}
  \hfill
  \begin{subfigure}{0.32\linewidth}
    \centering
    \includegraphics{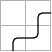}
    \hspace{1em}
    \includegraphics{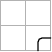}
    \caption{Bend segments}
  \end{subfigure}
  \caption{
    $2 \times 2$ units of T-shaped, L-shaped, and Neutral vertices for the embedding of \PlBiREMCC instances.
    Edge directions are omitted for simplicity.
  }
  \label{fig:zigzag_embedding}
\end{figure}

By the argument above, it suffices to construct a T-shaped vertex gadget and three L-shaped vertex gadgets (free, integral, biased) to realize the orientation model.
The overall framework is analogous to the Five Cells case described in \cref{sec:five_cells}.
Unlike the Five Cells case, these gadgets require a checkerboard arrangement of standard and inverted orientations to ensure consistent wire connections.
Regarding the weight constraints, each wire supplies a specific number of cells to the central green region (\cref{fig:fourcells_wire}), whose demand is set equal to the number of incident wires.
The remainder of the construction, including the wire variants (\cref{fig:fourcells_wire_variant}), follows the same principles as in the Five Cells case.

\begin{figure}[tbp]
  \centering
  \begin{subfigure}{0.32\linewidth}
    \centering
    \includegraphics[scale=\pdfscale]{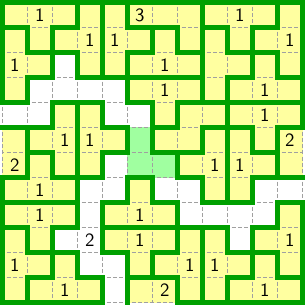}
    \caption{T-shaped}
  \end{subfigure}
  \hfill
  \begin{subfigure}{0.32\linewidth}
    \centering
    \includegraphics[scale=\pdfscale]{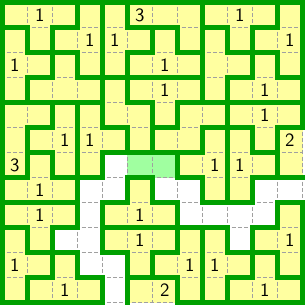}
    \caption{L-shaped}
  \end{subfigure}
  \hfill
  \begin{subfigure}{0.32\linewidth}
    \centering
    \includegraphics[scale=\pdfscale]{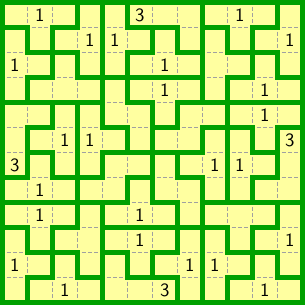}
    \caption{Neutral}
  \end{subfigure}
  \caption{
    Unit gadgets for Four Cells.
    The undetermined portion is divided into three (T-shaped) or two (L-shaped) wires.
    The green cells indicate the connections between the wires.
  }
  \label{fig:fourcells_unit_gadget}
\end{figure}

\begin{figure}[tbp]
  \centering
  \begin{subfigure}{0.32\linewidth}
    \centering
    \includegraphics[scale=\pdfscale]{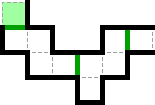}
    \caption{Weight $0$}
  \end{subfigure}
  \hfill
  \begin{subfigure}{0.32\linewidth}
    \centering
    \includegraphics[scale=\pdfscale]{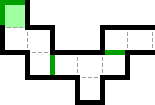}
    \caption{Weight $1$}
  \end{subfigure}
  \hfill
  \begin{subfigure}{0.32\linewidth}
    \centering
    \includegraphics[scale=\pdfscale]{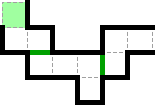}
    \caption{Weight $2$}
  \end{subfigure}
  \caption{
    Three local solutions of a wire.
    The number of cells in which the tetromino extends into the green cells corresponds to the weight contributed to the vertex.
  }
  \label{fig:fourcells_wire}
  \vspace{1em}
  \begin{subfigure}{0.48\linewidth}
    \centering
    \includegraphics[scale=\pdfscale]{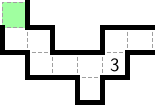}
    \caption{Incoming}
  \end{subfigure}
  \hfill
  \begin{subfigure}{0.48\linewidth}
    \centering
    \includegraphics[scale=\pdfscale]{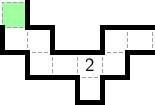}
    \caption{Exclusive}
  \end{subfigure}
  \caption{
    Two specific types of wires for Four Cells, each restricting the local solutions and thereby the admissible orientations.
  }
  \label{fig:fourcells_wire_variant}
\end{figure}

Because this reduction is parsimonious, we obtain the following result.

\begin{theorem}
  Four Cells is ASP-complete.
\end{theorem}

\FloatBarrier
\bibliography{references_arxiv}

\end{document}